\documentclass[aps,reprint,amsmath,amssymb]{revtex4-2}
\usepackage{amsthm}
\usepackage{graphicx}
\usepackage{dcolumn}
\usepackage{bm}
\usepackage[colorlinks=true, allcolors=blue]{hyperref}

 \newtheorem{lemma}{Lemma}

\newtheorem{cor}{Corollary}

\renewcommand{\(}{\left(}
\renewcommand{\)}{\right)}
\renewcommand{\[}{\left[}
\renewcommand{\]}{\right]}
\newcommand{\e}{\text{e}}
\newcommand{\Tr}{\text{Tr}}

\newcommand{\ad}{\text{ad}_H}

\begin{document}

\title{Trotter error timescaling separation via commutant decomposition}

\author{Yi-Hsiang Chen}
\affiliation{Quantinuum, 303 South Technology Court, Broomfield, Colorado 80021, USA}
\email{yihsiang.chen@quantinuum.com}


\begin{abstract}
Suppressing the Trotter error in dynamical quantum simulation typically requires running deeper circuits, posing a great challenge for noisy near-term quantum devices. Studies have shown that the empirical error is usually much smaller than the one suggested by existing bounds, implying the actual circuit cost required is much less than the ones based on those bounds. Here, we improve the estimate of the Trotter error over existing bounds, by introducing a general framework of commutant decomposition that separates disjoint error components that have fundamentally different scaling with time. In particular we identify two error components that each scale as $O(\tau^pt)$ and $O(\tau^p)$ for a $p$th-order product formula evolving to time $t$ using a fixed step size $\tau$, it implies one would scale linearly with time $t$ and the other would be constant of $t$. We show that this formalism not only straightforwardly reproduces previous results but also provides a better error estimate for higher-order product formulas. We demonstrate the improvement both analytically and numerically. We also apply the analysis to observable error relating to the heating in Floquet dynamics and thermalization, which is of independent interest.
\end{abstract}

\maketitle

\section{Introduction}
Simulating quantum dynamics has remained one of the most promising applications of quantum computers. Digitally performing a continuous-time evolution generated by a Hamiltonian $H$ with elementary quantum gates is known as Hamiltonian simulation. The earliest Hamiltonian simulation quantum algorithm was given by Lloyd \cite{Lloyd1996} where the global evolution is decomposed into primitive operations using Lie-Trotter formula \cite{trotter}. Numerous methods have since been developed to improve the asymptotic scaling in various parameters such as approximation error and time \cite{LCU2015,QSP2017}. However, the enormous prefactors in these algorithms render them impractical for near-term applications \cite{Childs2018}, whereas methods based on Lie-Trotter decomposition, known as product formulas (PFs)  \cite{SUZUKI1990319,Hatano2005}, and its variants \cite{Childs2019fasterquantum,low2019,Carrera2023,lubi2023} are known to have smaller constant overhead, hence leading to shorter circuits required for near-term devices. 

For an $n$-qubit system evolving to time $t$, the gate complexity of a fixed-accuracy $p$th-order PF is $O(n^k t^{1+1/p})$ \cite{PhysRevX.11.011020}, where $k$ is a model dependent constant. Its complexity in time $O(t^{1+1/p})$ is known to be suboptimal, although the optimality $O(t)$ can be \emph{practically} reached by increasing the order $p$ of the PF \cite{PhysRevLett.123.050503}. In addition, for systems with only two non-commuting terms, i.e., $H=H_1+H_2$, it was shown that the accumulated Trotter error scales as $O\(n\tau+n\tau^2t\)$ where $\tau$ is the step size and $t$ is the total evolution time \cite{PhysRevLett.124.220502,PhysRevLett.128.210501}. To simulate to an accuracy $\epsilon$, the gate complexity in time appears to be optimal $O(t)$ for short times even under a first-order PF \cite{PhysRevLett.124.220502,PhysRevLett.128.210501}. Tran \emph{et al.} \cite{PhysRevLett.124.220502} attributes this to the interference of the Trotter error between different steps whereas Layden \cite{PhysRevLett.128.210501} realizes it from the fact that the first-order PF has almost identical circuit structure as the second-order PF except the difference in the first and the last step, resulting in the $O(n\tau)$ term that dominates at short times. Remarkably, this $O(\tau)$ term not only is the reason behind the optimal timescaling but also implies the accumulated Trotter error does not increase with $t$ under a fixed step size $\tau$---a feature first observed in \cite{doi:10.1126/sciadv.aau8342}. However, it is unknown how the result generalizes beyond a two-term Hamiltonian with a first-order PF. 

In this paper, we focus on how Trotter error accumulates over time and provide a general framework to realize its behavior where we identify the key $O(\tau^p)$ term that has the optimal timescaling, for a general Hamiltonian $H$ simulated with a $p$th-order PF.  This is done by introducing a commutant decomposition for the error terms, where the component in the commutant scales as $O(\tau^pt)$ and the component in the orthogonal complement has $O(\tau^p)$. Under this formalism, it is clear that for a two-term Hamiltonian, the first-order PF has no component in the commutant (up to the leading order), hence recovering the $O(\tau+\tau^2t)$ scaling shown in \cite{PhysRevLett.124.220502,PhysRevLett.128.210501}. We also apply this decomposition for the second-order PF and obtain a better Trotter estimate than the previous bound \cite{PhysRevX.11.011020,Kivlichan2020improvedfault}. This can also be used to obtain a shorter circuit required for a given accuracy, where we observe an $\approx 40\% $ saving under a near-term relevant setting, hence making it more viable for near-term devices. Finally, we apply it to observable Trotter error and find implications for Floquet dynamics, heating and thermalization.

\section{timescaling separation}
Suppose $\e^{-iHt}$ is the ideal evolution up to time $t$ generated by the Hamiltonian $H$. The evolution is decomposed into $r$ steps using step size $\tau:=t/r$ where each step is 
\begin{align}
U:=\e^{-iH\tau}.
\end{align} 
The PF approximates each $U$ with small step size $\tau$ such that the accumulated evolution approaches the ideal evolution $\e^{-iHt}=U^r$. Denote the approximation error to $U$ as $\delta$, then the total Trotter error after $r$ steps is
\begin{align}
(U+\delta)^r-U^r=\left(\sum_{k=0}^{r-1}U^k \delta U^{-k}\right)U^{r-1} +O\left(r^2||\delta||^2\right), \label{totalerror1}
\end{align}
where we assume $r||\delta||<1$. See Appendix~\ref{app:U_delta} for details.  For a $p$th-order PF, $\delta$ is $O(\tau^{p+1})$ in step size $\tau$ \cite{PhysRevX.11.011020}. 
When $r||\delta||$ is small enough, the total Trotter error is dominated by the first term in Eq. (\ref{totalerror1}), which is a sum of different powers of the rotation, i.e.,
\begin{align}
\Delta :=\sum_{k=0}^{r-1}U^k \delta U^{-k}=\sum_{k=0}^{r-1}\e^{-ik\tau\ad}\delta,
\end{align}
where we denote the commutator as $\ad(\cdot):=[H,\cdot]$ and use the identity $\e^{-iHk\tau}(\cdot)\e^{iHk\tau}=\e^{-ik\tau\ad}(\cdot)$ with $\cdot$ denoting any operator. Now we introduce a commutant decomposition for the operator space, i.e., denoting $\mathcal{M}_{\mathbb{C}}(d,d)$ the set of all $d\times d$ matrices and $\mathcal{M}_{\mathbb{C}}(d,d)=\mathbf{H}_{||}\oplus\mathbf{H}_{\perp}$ where $\mathbf{H}_{||}=\{\forall M \in \mathcal{M}_{\mathbb{C}}(d,d)\ | \ [H,M]=0\}$ is the commutant of $H$ and $\mathbf{H}_{\perp}=\mathcal{M}_{\mathbb{C}}(d,d)\setminus \mathbf{H}_{||}$ is the orthogonal complement which is non-commuting with $H$. Then the error term $\delta$ can be uniquely decomposed as $\delta=\delta_{||}+\delta_{\perp}$, where $\delta_{||}\in \mathbf{H}_{||}$ and $\delta_{\perp}=\delta-\delta_{||}\in \mathbf{H}_{\perp}$ (see Appendix~\ref{app:commutant_deomp} for details). The accumulated error $\Delta$ becomes
\begin{align}
\Delta=r\delta_{||}+\sum_{k=0}^{r-1}\e^{-ik\tau\ad}\delta_{\perp}. \label{totalerror2}
\end{align}
One observes that different components of $\delta$ accumulate differently, i.e., the component that commutes with $H$ adds linearly in steps while the non-commuting part gets rotated at every step then added. Fig. \ref{rotation_plot} illustrates this effect where $H$ rotates $\delta_{\perp}$ by the same amount for each step while $\delta_{||}$ is ``parallel" (commuting) to the rotation axis $H$ and unaffected. 
\begin{figure}[h]
\centering
\includegraphics[width=4cm]{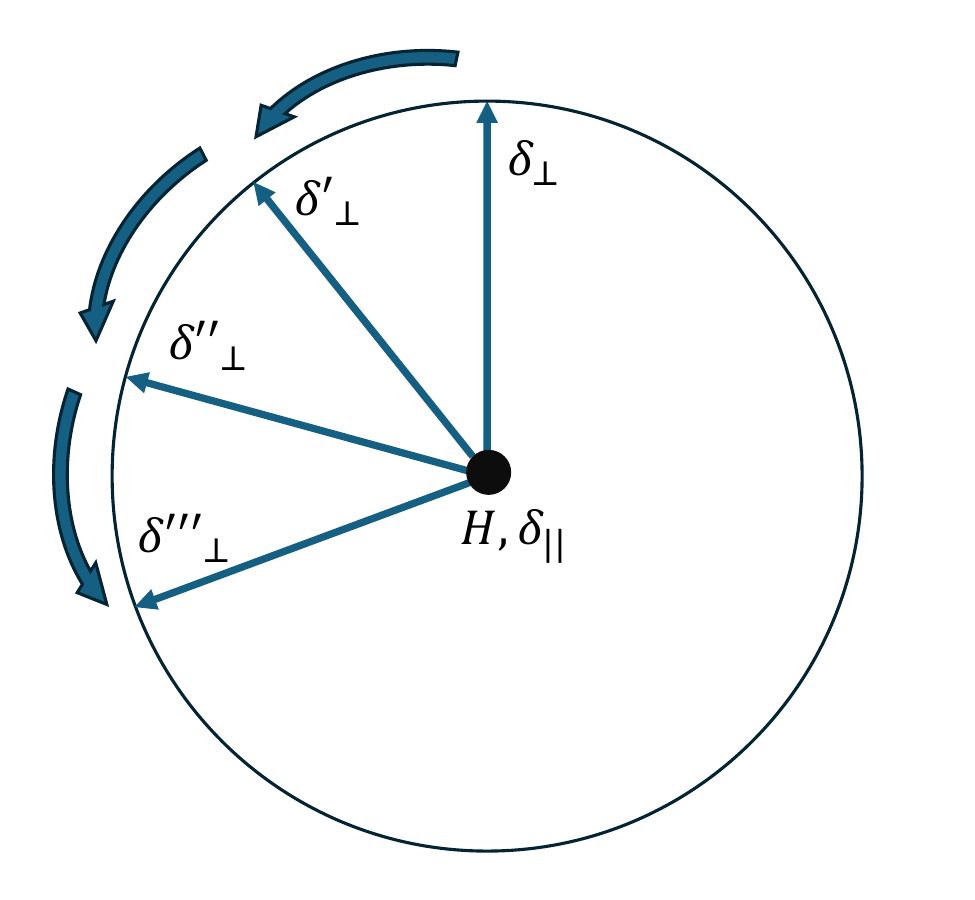}
\caption{$\delta_{\perp}$ is rotated by $H$ pointing outward from the plane and $\delta_{||}$ is ``parallel'' to $H$ and unaffected.}
\label{rotation_plot}
\end{figure}
This is the key effect that leads to different error scalings in a Trotterized circuit. To further simplify $\Delta$, one can rewrite the sum of rotation operators as \footnote{Observe that this is a geometric series applied to the operator $\ad$, which is a diagonal matrix in $H$'s eigen-operator representation $|i\rangle\langle j|$ (Appendix~\ref{app:commutant_deomp}). The geometric series is then individually applied to each diagonal element.}
\begin{align}
\sum_{k=0}^{r-1}\e^{-ik\tau\ad}=\frac{1-\e^{-it\ad}}{1-\e^{-i\tau\ad}}, \label{geometricsum}
\end{align}
where the operator $(1-\e^{-i\tau\ad})^{-1}$ is well defined only on $\delta_{\perp}$ (but not $\delta_{||}$). This inverse operator can be further expanded in terms of $\ad^{-1}$, which is also well defined only for elements in $\mathbf{H}_{\perp}$ and is the inverse operation of $\ad$ satisfying $\ad^{-1}\ad(\delta_{\perp})=\delta_{\perp}$. Therefore,
\begin{align}
(1-\e^{-i\tau\ad})^{-1}=\frac{-i}{\tau}\ad^{-1}+O(1), \label{inv_expansion}
\end{align}
which is derived in Appendix \ref{app:for_eq6}. The reason to keep it to order $O(1/\tau)$ is it acts on the operator $\delta_{\perp}$ that is $O(\tau^{p+1})$. Substituting Eq. (\ref{inv_expansion}) to Eq. (\ref{geometricsum}) where the geometric series is replaced by $1-\e^{-it\ad} $ multiplying Eq. \eqref{inv_expansion}, and applying the series to $\delta_{\perp}$ in Eq. (\ref{totalerror2}), we have
\begin{align}
\Delta=&r\delta_{||}+\frac{i}{\tau}\left[\e^{-iHt}\ad^{-1}(\delta_{\perp})\e^{iHt}-\ad^{-1}(\delta_{\perp})\right]  \nonumber\\
&+O\left(\tau^{p+1}\right).
\end{align}
Therefore, the total Trotter error in the unitary is
\begin{align}
&\left|\left|(U+\delta)^r-U^r\right|\right|=||\Delta|| +O(r^2||\delta||^2) \nonumber\\
&=\left|\left| r\delta_{||}+\frac{i}{\tau}\left[\e^{-iHt}\ad^{-1}(\delta_{\perp})\e^{iHt}-\ad^{-1}(\delta_{\perp})\right] \right|\right| \nonumber\\
&\ \ \ +O\left(\tau^{p+1}\right) \nonumber\\
&\leq r\left|\left|\delta_{||}\right|\right|+  \frac{2}{\tau}\left|\left|\ad^{-1}(\delta_{\perp})\right|\right| +O\left(\tau^{p+1}\right)\label{errorbound}\\
&= O\left(\tau^pt +\tau^p\right) \label{errorscaling}.
\end{align}
This shows a separation of the error scaling with $t$ in a $p$th-order PF. In particular, given a fixed step size $\tau$, the component from $\delta_{\perp}$ does not scale with $t$ while the part from $\delta_{||}$ scales linearly with $t$. To simulate the dynamics with a fixed accuracy $\epsilon$, one would use large enough $r$ such that the total error is below $\epsilon$. The total gate complexity is then proportional to $r$. Eqs. (\ref{errorbound}) and (\ref{errorscaling}) indicate the complexity in time will depend on the relative size of $\delta_{||}$ and $\delta_{\perp}$. When $||\delta_{||}|| \gg ||\delta_{\perp}||$, it approaches $O(t^{1+1/p})$ as in previous results \cite{PhysRevX.11.011020} while it saturates the lower bound $O(t)$ \cite{doi:10.1137/18M1231511} when $||\delta_{||}|| \ll ||\delta_{\perp}||$. 

\subsection{Two-term Hamiltonians}
Now we consider the generic case where the Hamiltonian is a sum of only two non-commuting parts, i.e., $H=H_1+H_2$, where each $H_{1,2}$ is a sum of commuting Paulis. Examples of such systems include all one-dimensional (1D) nearest-neighbor spin chains and the transverse-field Ising (TFI) model in any lattice dimension. 

\subsubsection{Re-deriving first-order PF's result}
It is shown in \cite{PhysRevLett.124.220502,PhysRevLett.128.210501} that the Trotter error of the first-order PF scales as $O(\tau+\tau^2t)$, which is a tighter bound than one would naively obtain using $O(\tau^2)$ for each step and $O(\tau t)$ being the total error after $r=t/\tau$ steps. Here, we show that this can be straightforwardly seen via commutant decomposition in Eq. (\ref{errorbound}). Recall that for each step, the error in the first-order PF is  
\begin{align}
\delta&=\frac{1}{2}[H_1,H_2]\tau^2+O(\tau^3)\nonumber\\
&=\frac{1}{2}[H,H_2]\tau^2+O(\tau^3)= \delta_{||}+\delta_{\perp},
\end{align}
where 
\begin{align}
\delta_{||}=0+O(\tau^3),\ \ \delta_{\perp}=\frac{1}{2}[H,H_2]\tau^2+O(\tau^3). \label{deltaterms}
\end{align}
Note that $\delta_{||}=0$ up to order $O(\tau^2)$ because $[H,H_2]$ has no component in the commuting space $\mathbf{H}_{||}$ as shown in Corollary \ref{ad_in_Hperp} in Appendix~\ref{app:commutant_deomp}. Substituting Eq. (\ref{deltaterms}) to Eq. (\ref{errorbound}), we see that the Trotter error has
\begin{align}
&\left|\left|(U+\delta)^r-U^r\right|\right|\leq\tau\left|\left| \ad^{-1}(\ad(H_2))\right|\right|+O\left(r\tau^3\right) \nonumber \\
&=\tau\left|\left| H_{2\perp}\right|\right|+O\left(\tau^2t\right)=O\(\tau+\tau^2t\),
\end{align}
which recovers the timescaling obtained in \cite{PhysRevLett.124.220502,PhysRevLett.128.210501}.

\subsubsection{Second-order PF}
As explained in \cite{PhysRevLett.128.210501}, for Hamiltonian $H=H_1+H_2$, the circuit structure of the first- and second-order PF is nearly identical except the difference at the beginning and the end, implying one should use the second-order formula, i.e.,
\begin{align}
U_2=\e^{-iH_2\tau/2}\e^{-iH_1\tau}\e^{-iH_2\tau/2}. \nonumber
\end{align} 
Previous studies \cite{Kivlichan2020improvedfault,PhysRevX.11.011020} have shown that the cumulative Trotter error $U^r_2$ can be bounded by 
\begin{align}
&\left|\left|U_2^r-U^r\right|\right| \label{existingbound} \\
&\leq \(\frac{1}{12}\left|\left|[H_1,[H_1,H_2]]\right|\right|+\frac{1}{24}\left|\left|[H_2,[H_2,H_1]]\right|\right|\)\tau^2t. \nonumber
\end{align}
This bound is strictly proportional to $\tau^2t$ which is linear in $t$ under a fixed step size $\tau$. In the following, we show one can extract components proportional to $\tau^2$ and obtain a tighter Trotter error estimate for $U_2^r$. First note that the error per step is (see Appendix~\ref{PF2_error_term})
\begin{align}
\delta=\(\frac{i}{24}[H_1,[H_1,H_2]]+\frac{i}{24}[H,[H,H_2]]\)\tau^3 +O\left(\tau^4\right). \label{errorterm}
\end{align}
Observe that the term $[H,[H,H_2]]$ belongs to $\mathbf{H}_{\perp}$ and its cumulative effect is a rotation that has $1/\tau$ dependence. The term $[H_1,[H_1,H_2]]$ can have components in both $\mathbf{H}_{||}$ and $\mathbf{H}_{\perp}$. Ideally one would also separate out its $\mathbf{H}_{\perp}$ part but it is not straightforward for general $H_1$ and $H_2$. Nevertheless, we can bound $\Delta$ by identifying $[H,[H,H_2]]\in \mathbf{H}_{\perp}$ (as shown in Corollary \ref{ad_in_Hperp} in Appendix \ref{app:commutant_deomp}) which leads to $\propto 1/\tau$ while bounding the $[H_1,[H_1,H_2]]$ term with triangular inequality leading to $\propto r$ scaling, i.e.,
\begin{align}
&||\Delta||=\left|\left| \sum_{k=0}^{r-1} U^k\delta U^{-k} \right|\right|\nonumber\\
&\leq \frac{1}{24}\left|\left|\sum_{k=0}^{r-1}U^{k}[H_1,[H_1,H_2]] U^{-k}\right|\right|\tau^3 \nonumber \\
&\ \ \ + \frac{2}{\tau}\frac{1}{24}\left|\left|\ad^{-1}\([H,[H,H_2]]\)\right|\right|\tau^3+O\(\tau^3t\)\\
&\leq \frac{1}{24}\left|\left|[H_1,[H_1,H_2]]\right|\right|\tau^2t +\frac{1}{12}\left|\left|[H_1,H_2]\right|\right|\tau^2+O\(\tau^3t\), \label{tighterbound}
\end{align}
where in the last inequality, we use the fact that $\ad^{-1}(\ad(O_{\perp}))=O_{\perp}$ for any operator $O_{\perp}$ in $\mathbf{H}_{\perp}$ and taking $O_{\perp}=[H,H_2]$ (see Appendix \ref{app:commutant_deomp}). Therefore we have
\begin{align}
&\left|\left|U_2^r-U^r\right|\right|\approx ||\Delta|| \nonumber\\
&\lesssim \frac{1}{24}\left|\left|[H_1,[H_1,H_2]]\right|\right|\tau^2t+\frac{1}{12}\left|\left|[H_1,H_2]\right|\right|\tau^2. \label{tighterbound1}
\end{align}
Notice that the above has terms that scale separately as $O\left(\tau^2t\right)$ and $O(\tau^2)$ while Eq. (\ref{existingbound}) is completely $O(\tau^2t)$. This is the main reason that Eq. (\ref{tighterbound1}) provides a better error estimate than the previous bound Eq. (\ref{existingbound}). We demonstrate this improvement numerically by applying it to two systems--- the 1D mixed-field Ising (MFI) model and the two-dimensional (2D) TFI model. 
\begin{figure}[h]
\centering
\includegraphics[width=4.2cm]{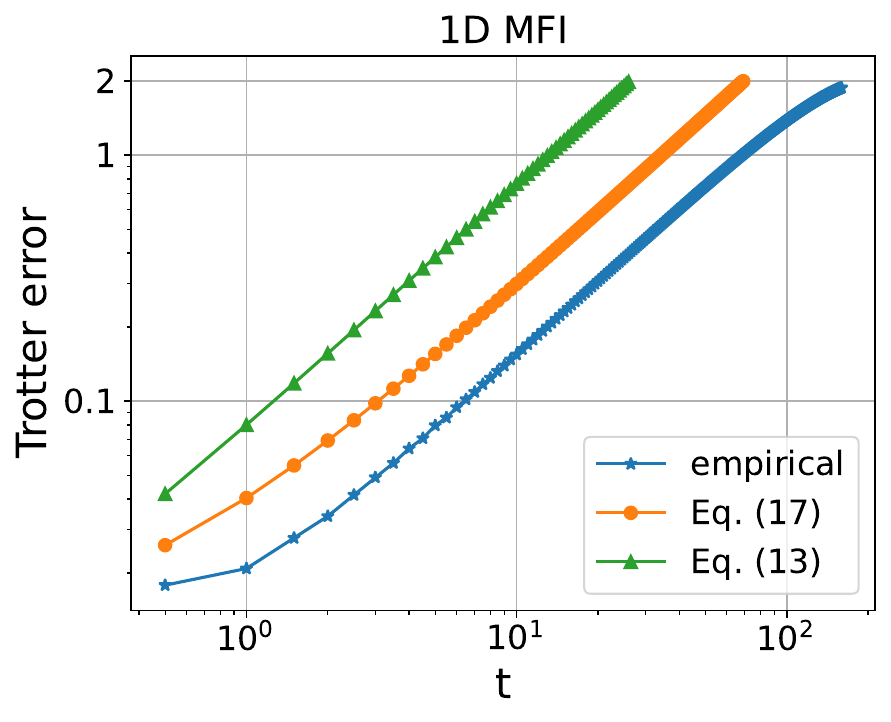}
\includegraphics[width=4.2cm]{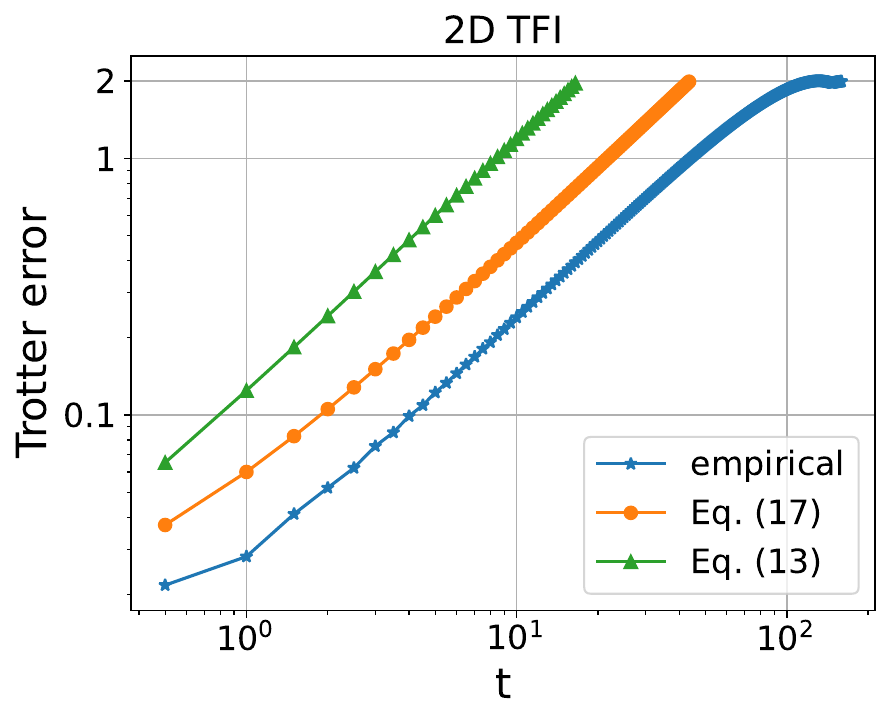}
\caption{The comparison between the previous bound Eq. (\ref{existingbound}) and Eq. (\ref{tighterbound1}) with empirical Trotter error $\left|\left|U^r_2-U^r\right|\right|$, evolving under a fixed step size $\tau=0.05$. $||\cdot||$ is the spectral norm. For 1D MFI, $H_1=J\sum_{i=1}^{11}Z_iZ_{i+1}+h_z\sum_{i=1}^{12} Z_i$ and $H_2=h_x\sum_{i=1}^{12} X_i$ with $J=1, h_z=1.1,$ and $h_x=1.07$. For 2D TFI, $H_1=J\sum_{<i,j>}Z_iZ_j$ where the sum is over nearest-neighbor couplings on a $3\times4$ lattice with periodic boundary and $H_2=h\sum_{i=1}^{12}X_i$ with $J=1$ and $h=1.13$.}
\label{boundplot}
\end{figure}
Fig. \ref{boundplot} compares Eq. (\ref{tighterbound1}) and Eq. (\ref{existingbound}) with the true Trotter error $||U^r_2-U||$ using a fixed step size $\tau=0.05$ evolving up to $t>100$. This shows that one can obtain a tighter estimate by extracting the error components in $\mathbf{H}_{\perp}$ that have a constant-in-time scaling. This can also be used to obtain a shorter circuit depth required to simulate to a given accuracy. For example, to simulate the 1D MFI with 400 qubits up to $t=400$ with accuracy 0.01 (in terms of system size normalized Frobenius norm $||\cdot||_{F}/\sqrt{d}$), the previous bound Eq. (\ref{existingbound}) implies 545244 steps required while Eq. (\ref{tighterbound1}) indicates 334254 steps needed \footnote{Setting the right-hand side of Eq.~\eqref{existingbound} as 0.01 for the required total accuracy, and using $\tau=t/r$ with total time $t=400$, we deduce the required number of steps is $\sqrt{\frac{t^3}{0.01}\(\frac{1}{12}\left|\left|[H_1,[H_1,H_2]]\right|\right|+\frac{1}{24}\left|\left|[H_2,[H_2,H_1]]\right|\right|\)},$ resulting in 545244 steps required using existing bound Eq.~\eqref{existingbound}. Similarly, setting the right-hand side of Eq.~\eqref{tighterbound1} as 0.01 for the required accuracy, we deduce the required steps as $\sqrt{\frac{1}{0.01}\(\frac{1}{24}\left|\left|[H_1,[H_1,H_2]]\right|\right|t^3+\frac{1}{12}\left|\left|[H_1,H_2]\right|\right|t^2\)}$}. Although the method is demonstrated for the first- and second-order PF in the above, the same idea can apply to the higher-order PF by extracting error terms corresponding to $\mathbf{H}_{\perp}$ (e.g., by identifying terms like $[H,\cdot]$) that lead to $O(\tau^p)$ scaling.

\section{Application to observable error}
Now we discuss the Trotter error in observables. Given an initial density matrix $\rho$ and an observable $O$, the observable Trotter is defined as (see Appendix \ref{app:obs_trot_err})
\begin{align}
&\triangle\langle O\rangle_{\rho}:=\Tr\left[(U+\delta)^{-r}O(U+\delta)^r\rho-U^{-r}OU^r\rho\right] \nonumber\\
&=\Tr\left[(\Delta^{\dagger}O+O\Delta)(U^r\rho U^{-r})\right]+ O\left(\tau^{p+1}t\right).
\end{align}
For simplicity, we consider again a two-term Hamiltonian $H=H_1+H_2$ with a second-order PF. The error per step can be written as a sum of two anti-Hermitian operators each corresponding to components in $\mathbf{H}_{||}$ and $\mathbf{H}_{\perp}$, i.e., $\delta\approx i(H_{||}+H_{\perp})\tau^3$, where $H_{||}\in \mathbf{H}_{||}$ and $H_{\perp} \in \mathbf{H}_{\perp}$ are Hermitian. Then one obtains observable error as (see Appendix \ref{app:obs_trot_err})
\begin{align}
&\triangle\langle O\rangle_{\rho} \approx i\Tr\left\{\left[O,H_{||}\right]\rho(t)\right\}\tau^2 t\label{obs_err} \\
&\ \ \ \ +\Tr\left\{\left[O,\ad^{-1}(H_{\perp})-\e^{-iHt}\ad^{-1}(H_{\perp})\e^{iHt}\right]\rho(t) \right\}\tau^2,\nonumber
\end{align}
where $\tau=t/r$ is the step size and $\rho(t)= \e^{-iHt}\rho \e^{iHt}$. In the following, we focus on the fixed-step-size case, i.e., $\tau$ is a constant. We call $O(\tau^2t)$ and $O(\tau^2)$ in Eq. (\ref{obs_err}) the linear-in-$t$ and the constant-in-$t$ terms separately \footnote{To clarify, the value of the $O(\tau^2)$ term is a complicated combination of many oscillations and it could look linear in a particular period under certain circumstances, but we call it constant-in-$t$ in the sense that its overall effect is oscillation that does not grow with time asymptotically}. It appears that for any observable that is a function of $H$, i.e., $O=f(H)$, the linear-in-$t$ component vanishes, since $[f(H),H_{||}]=0$. Therefore, under a fixed step size $\tau$, the Trotter error of $f(H)$ is constant in $t$, with local oscillation from the difference between $\ad^{-1}(H_{\perp})$ and $\e^{-iHt}\ad^{-1}(H_{\perp})\e^{iHt}$. 
\begin{figure}[h]
\centering
\includegraphics[width=4.25cm]{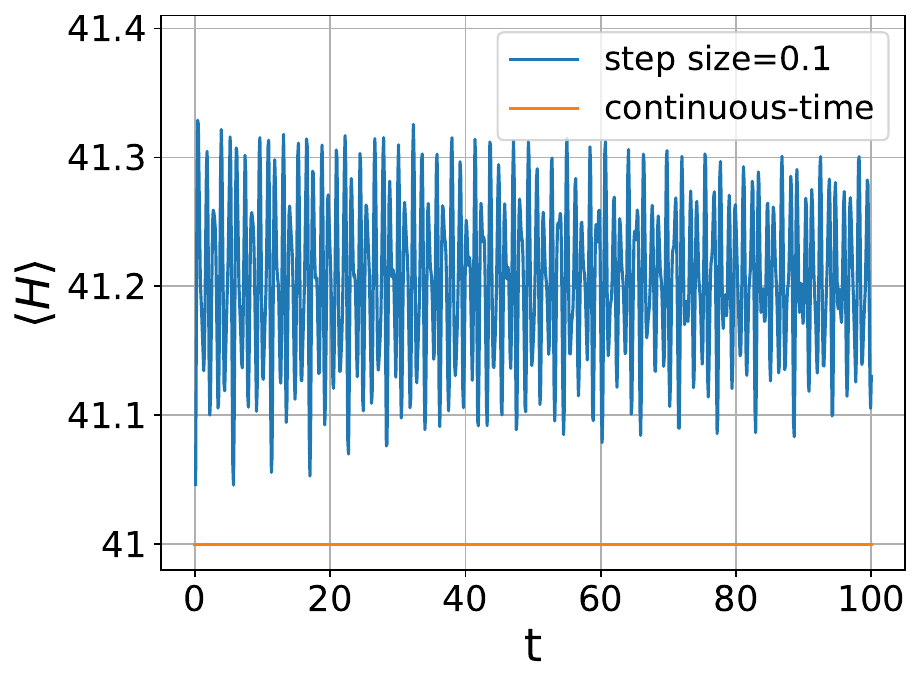}
\includegraphics[width=4.25cm]{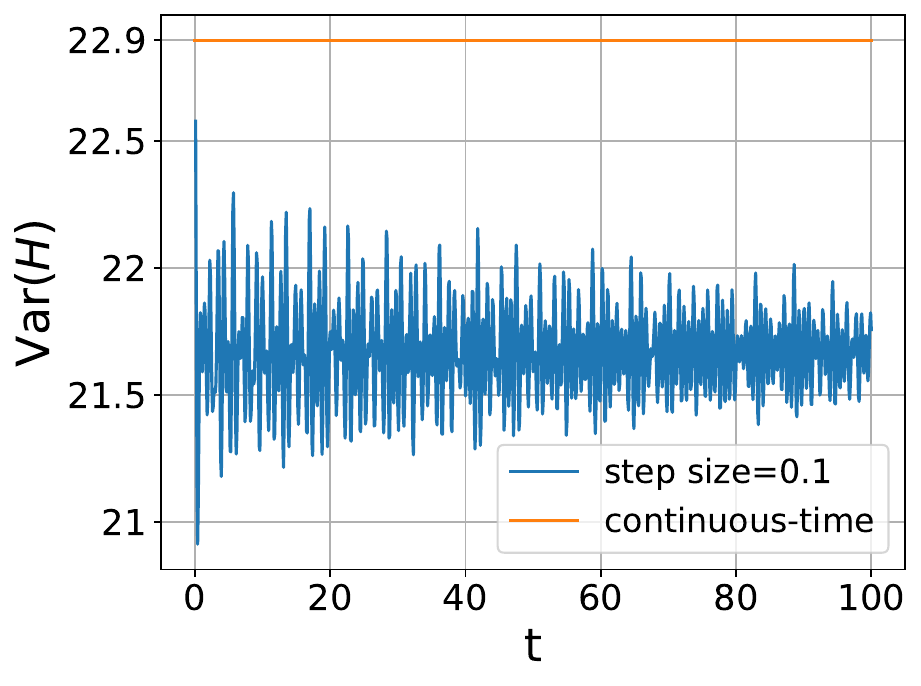}
\caption{The evolution of the energy $\langle H\rangle$ (left plot) and the energy variance $\text{Var}(H)=\langle H^2\rangle -\langle H\rangle^2$ (right plot) using a step size $\tau=0.1$ and comparison to the continuous-time value. The system is a 20-qubit 1D MFI with $J=1, h_z=1.1,$ and $h_x=1.07$.}
\label{Hplot}
\end{figure}
Fig. \ref{Hplot} illustrates this behavior for the energy $\langle H\rangle$  and the energy variance $\text{Var}(H)=\langle H^2\rangle -\langle H\rangle^2$ evolving under a fixed step size $\tau=0.1$. Conserved quantities in a continuous-time evolution become non-conserved due to the discretization. However the remarkable fact is the deviation does not appear to \emph{grow} with time $t$. In the context of Floquet dynamics \cite{WeidingerKnap2017,HO2023169297,KUWAHARA201696,PhysRevLett.116.120401} where one realizes a Trotterized evolution with a finite step size as a time-dependent periodically driven system, the Trotter error in energy is equivalent to the Floquet heating. It has been shown \cite{PhysRevLett.116.120401} that the heating timescale is exponentially long in the inverse of the step size. Here we obtain a parallel but inequivalent result that says there is no heating up to a time scale proportional to the inverse of the square of the step size \footnote{The analysis holds well for $r||\delta||\ll1$. Since $r||\delta||=O(\tau^2t)$ for the second-order PF, under a fixed step size $\tau=t/r$, it implies $t\ll \frac{1}{\tau^2}$.}, in a second-order PF. 

For local observables that thermalize under the eigenstate thermalization hypothesis (ETH), the expectation value is determined by the so-called diagonal ensemble, i.e., $\langle O\rangle_{\rho(t)}=\sum_{i} \rho_{ii} \langle i|O|i\rangle $ where $|i\rangle$ is the energy eigenbasis of $H$ and $\rho_{ii}=\langle i|\rho|i\rangle$ are the diagonal elements of the initial density matrix $\rho$. This amounts to replacing $\rho(t)\to \sum_{i} \rho_{ii}|i\rangle\langle i|$. Substituting this into Eq. (\ref{obs_err}), we find that the linear-in-$t$ term vanishes because $H_{||}$ is also diagonal in the energy eigenbasis for non-degenerate $H$ as shown in Appendix \ref{app:commutant_deomp}. and $\Tr\{[O,H_{||}]\rho(t)\}= \sum_{i} \rho_{ii} \Tr\{[H_{||},|i\rangle\langle i|]O\}=0$. This agrees with the result in \cite{doi:10.1126/sciadv.aau8342} that the Trotter error in an ETH observable does not grow with time.

\section{Conclusion and discussion}
In summary, we provided a general framework of commutant decomposition to understand the accumulation of Trotter error over time. This generalizes the concept of error interference introduced in \cite{PhysRevLett.124.220502}, i.e., for a general $H$ under a $p$th-order product formula, there are two error components that each scale as $O(t^{p+1}/r^p)$ and $O(t^p/r^p)$. The component with $O(t^{p+1}/r^p)$ adds linearly with the number of steps with no destructive interference whereas the component with $O(t^p/r^p)$ adds rotationally such that there is no net increase. The previous result \cite{PhysRevLett.124.220502,PhysRevLett.128.210501} for $H=H_1+H_2$ with the first-order product formula is then a special case where the $O(t^{2}/r)$ component vanishes. In fact, when the $O(t^{p+1}/r^p)$ term is present, it is \emph{necessary} for a product formula to scale as $O(t^{1+1/p})$ in gate complexity and hence a lower bound. Even though the commutant decomposition does not change the asymptotic scaling with time, one can use it to extract the component with $O(t^p/r^p)$ to improve the error estimate, resulting in less gate cost required in practice. We provided such an example for the second-order product formula. We note that exactly performing the commutant decomposition can be challenging for system sizes beyond exact diagonalization. However, we have shown that one can extract components in $\mathbf{H}_{\perp}$ by identifying terms in the form of $[H,\cdot]$. Exploring methods for the decomposition for more cases remains an open question. We also note that our analysis on the heating in Floquet dynamics is for a second-order product formula where the no-heating effect lasts up to $O(1/\tau^2)$. It may suggest that the no-heating time scale is $O(1/\tau^p)$ for a $p$th-order product formula, which we leave as an open question.

\begin{acknowledgments}
We thank Eli Chertkov and Michael Foss-Feig for valuable discussions and Charlie Baldwin for helpful feedback.
\end{acknowledgments}


\appendix

\setcounter{section}{0}
\setcounter{equation}{0}

\section{Analysis of $(U+\delta)^r-U^r$} \label{app:U_delta}
We note that this is also the starting point of the derivation in \cite{PhysRevLett.124.220502}. Here we give a self-contained explanation for Eq.~\eqref{totalerror1} in the paper, especially showing the higher order terms ignored are indeed $O(r^2||\delta||^2)$. 

First of all, by keeping track of the terms that involve one $\delta$ and no $\delta$ in the expansion $(U+\delta)^r$, we have
\begin{align}
(U+\delta)^r-U^r-\left(\sum_{k=0}^{r-1}U^k \delta U^{-k}\right)U^{r-1}=\sum_{l=2}^{r}\sum_{X\in \mathcal{X}_l}X,
\end{align}
where $\mathcal{X}_l$ is the set of all possible operators involving $l$ $\delta$s, e.g., terms like $X=U\cdots U\delta U\cdots U\delta U\cdots U\in \mathcal{X}_2$. The norm of the above remainder is (by moving the norm inside the sum and using $||U||=1$)
\begin{align}
&\left|\left| (U+\delta)^r-U^r-\left(\sum_{k=0}^{r-1}U^k \delta U^{-k}\right)U^{r-1}\right|\right| \nonumber \\
&\leq \sum_{l=2}^{r} {r \choose l}||\delta||^l \nonumber\\
&=\sum_{l=2}^{r} \frac{r(r-1)\cdots (r-l+1)}{l!}||\delta||^l \nonumber \\
&\leq \sum_{l=2}^{r} \frac{(r||\delta||)^l}{l!}\leq \sum^{\infty}_{l=2} \frac{(r||\delta||)^l}{l!} \leq \frac{\e}{2} r^2||\delta||^2,
\end{align}
where in the last inequality we use Taylor's remainder for the exponential and the assumption $r||\delta||<1$.

\section{Constructive description for the commutant of $H$}\label{app:commutant_deomp}
Here we provide an explicit representation for the commutant of $H$ and its orthogonal complement. We start by noting the following lemma.
\begin{lemma}
Let $H$ be an $d\times d$ Hermitian matrix. Denote $\mathcal{M}_{\mathbb{C}}(d,d)$ the set of all complex-valued $d\times d$ matrices. Define the commutant of H as 
\begin{align}
\mathbf{H}_{||}=\{\forall M \in \mathcal{M}_{\mathbb{C}}(d,d)\ | \ [H,M]=0\},
\end{align} 
i.e., the set of all $d\times d$ matrices that commute with $H$. Then $\mathbf{H}_{||}$ is equal to the space spanned by the eigenprojectors of $H$ if and only if $H$ is non-degenerate.
\end{lemma}

\begin{proof}
Since $H$ is Hermitian, its eigenvectors form an orthonormal basis $\{|i\rangle\}$ and any matrix can be represented in this basis, i.e., $M=\sum_{i,j} m_{ij}|i\rangle\langle j|$. Suppose $\mathbf{H}_{||}$ is equal to the space spanned by the eigenprojectors of $H$, i.e., $\mathbf{H}_{||}=\text{Span}\{|i\rangle\langle i|\}$, but there is a degeneracy, i.e., $H|i\rangle=\lambda|i\rangle$ and $H|j\rangle=\lambda|j\rangle$ with $i\neq j$, then consider a matrix $M=|i\rangle\langle j|$,and we have
\begin{align}
[H,M]=HM-MH=(\lambda-\lambda)|i\rangle\langle j|=0.
\end{align}
This contradicts with the initial assumption that $\mathbf{H}_{||}=\text{Span}\{|i\rangle\langle i|\}$.

Conversely, suppose $H$ is non-degenerate. It is clear that $\text{Span}\{|i\rangle\langle i|\}\subset\mathbf{H}_{||}$, since they are diagonal in the same basis as $H$. One only needs to show that for any matrix $M=\sum_{i\neq j}m_{ij} |i\rangle\langle j|\neq 0$, $[H,M]$ is not zero. Indeed,
\begin{align}
[H,M]=\sum_{i\neq j} (\lambda_i-\lambda_j) m_{ij} |i\rangle\langle j| \neq 0,
\end{align}
because $\lambda_i\neq \lambda_j$ for any $i \neq j$. This completes the proof.
\end{proof}

Denote the orthogonal complement of $\mathbf{H}_{||}$ as $\mathbf{H}_{\perp}:=\mathcal{M}_{\mathbb{C}}(d,d)\setminus \mathbf{H}_{||}$. Decomposing any operator $O$ into the commutant component $O_{||}\in  \mathbf{H}_{||}$ and the orthogonal complement $O_{\perp} \in \mathbf{H}_{\perp}$ can be explicitly done by working in the eigenbasis $|i\rangle$ of the Hamiltonian $H$. The procedure is as follows.
\begin{enumerate}
\item Find the eigenspectrum of $H$, i.e., $\{\lambda_i, |i\rangle\}$ the eigenvalues and eigenvectors.
\item Write $O$ in the basis of $|i\rangle$, i.e., $O_{ij}=\langle i|O|j\rangle$.
\item If $H$ is non-degenerate, then $O_{||}=\sum_{i=1}^{d} O_{ii} |i\rangle\langle i|$ is the commutant component of $O$, and $O_{\perp}=O-O_{||}$ is the rest.
\item If $H$ is degenerate, then $O_{||}=\sum_{i=1}^{d}O_{ii} |i\rangle\langle i| + \sum_{(i,j)\in \text{ same deg. space}} O_{ij} |i\rangle\langle j| $, and $O_{\perp}=O-O_{||}$. The sum $\sum_{(i,j)\in \text{ same deg. space}}$ is over all pairs of $i,j$ such that $ |i\rangle$ and $|j\rangle$ correspond to the same eigenvalue.
\end{enumerate}
The above construction directly implies the following corollaries.
\begin{cor}\label{Hperp_rep}
Given the eigenbasis $\{|i\rangle\}$ and the corresponding eigenvalues $\{\lambda_i\}$ of a Hermitian operator $H$, the orthogonal complement $\mathbf{H}_{\perp}$ is
\begin{align}
\mathbf{H}_{\perp}=\text{Span}\{|i\rangle\langle j|\ |\  \lambda_i\neq \lambda_j\}.
\end{align}
\end{cor}
The orthogonal complement space is spanned by the projectors $|i\rangle \langle j|$ such that $i$ and $j$ correspond to different eigenspaces $\lambda_i\neq\lambda_j$.
\begin{cor}\label{ad_in_Hperp}
The action of the commutator $\ad(\cdot)=[H,\cdot]$ on any operator $O$ can be expressed as 
\begin{align}
\ad(O)=[H,O]=\sum_{\lambda_i\neq\lambda_j} (\lambda_i-\lambda_j) O_{ij} |i\rangle\langle j|\in\mathbf{H}_{\perp},
\end{align}
which is completely in the orthogonal complement $\mathbf{H}_{\perp}$. In addition, the inverse action of $\ad$, i.e., $\ad^{-1}$, acts only on the orthogonal component $O_{\perp}\in\mathbf{H}_{\perp}$, and it can be expressed as 
\begin{align}
\ad^{-1}(O_{\perp})=\sum_{\lambda_i\neq\lambda_j} \frac{1}{\lambda_{i}-\lambda_j} O_{ij} |i\rangle\langle j|.
\end{align}
\end{cor}
It can be easily checked that $\ad^{-1}(\ad(O))=O_{\perp}$ and $\ad^{-1}(\ad(O_{\perp}))=\ad(\ad^{-1}(O_{\perp}))=O_{\perp}$. In addition, $\ad$ is a diagonal matrix in the presentation of $H$'s eigenprojectors $|i\rangle\langle j|$ with diagonal elements being $\lambda_i-\lambda_j$ and $\ad^{-1}$ is diagonal with elements $1/(\lambda_i-\lambda_j)$ but only defined on $\mathbf{H}_{\perp}$.

\section{Derivation of Eq.~\eqref{inv_expansion}}\label{app:for_eq6}
Here we justified the expansion in Eq.~\eqref{inv_expansion} in the main text. First note that from Corollary \ref{Hperp_rep}. and \ref{ad_in_Hperp}. in Appendix \ref{app:commutant_deomp}, we can view $\ad$ as a diagonal matrix and $(1-\e^{-i\frac{t}{r}\ad})^{-1}$ amounts to finding the inverse of each diagonal element. To simplified the notation, for a diagonal element $x$ of $\ad$, the diagonal element of $(1-\e^{-i\frac{t}{r}\ad})^{-1}$ is 
\begin{align}
\frac{1}{1-\e^{-i\tau x}},
\end{align}
where $\tau=t/r$ is the step size that is small. Therefore, we evaluate 
\begin{align}
\frac{1}{\tau}\(\frac{\tau}{1-\e^{-i\tau x}}\)=\frac{1}{\tau}\[\frac{-i}{x}+\frac{1}{2}\tau+\frac{ix}{12}\tau^2+O\(\tau^3\)\].
\end{align}
The constant order term in the bracket is proportional to the diagonal element of $\ad^{-1}$, i.e., simply the inverse of the corresponding diagonal element of $\ad$. Therefore,
\begin{align}
(1-\e^{-i\tau\ad})^{-1}=\frac{1}{\tau}\[-i\ad^{-1}+O\(\tau\)\],
\end{align}
which implies Eq.\eqref{inv_expansion}. The reason to have an expansion with leading order $1/\tau$ is this operator will act on error terms of order $\tau^{p+1}$ such that the final value is $O\(\tau^p\)$ which is small.

\section{Leading error terms in 2nd order product formula}\label{PF2_error_term}
Suppose the Hamiltonian is a sum of two non-commuting blocks, i.e., 
\begin{align}
H=H_1+H_2,
\end{align}
then the second-order product formula has leading order terms that are third order in step size. To see this, we first define functions
\begin{align}
&U_{prod}(\tau)=\e^{-iH_2\tau/2}\e^{-iH_1\tau}\e^{-iH_2\tau/2} \\
&U(\tau)=\e^{-i(H_1+H_2)\tau} \\
&\delta(\tau)=U_{prod}(\tau)-U(\tau).
\end{align}
$\delta$ is the Trotter error per step and its derivatives evaluated at $\tau=0$ give the expansion of Trotter error in the step size $\tau$, i.e.,
\begin{align}
\frac{d}{d\tau}\delta(\tau)&=\frac{-i}{2}H_2U_{prod}(\tau)+\e^{-i[H_2,\cdot]\tau/2}(-iH_1)U_{prod}(\tau) \nonumber\\
&\ \ \ \ +U_{prod}(\tau)\frac{-i}{2}H_2+i(H_1+H_2)U(\tau) \nonumber\\
&\implies \frac{d}{d\tau} \delta(\tau)\Big|_{\tau=0}=0. 
\end{align}

\begin{align}
\frac{d^2}{d\tau^2}\delta(\tau)&= \frac{-i}{2}H_2U'_{prod}+\frac{-i}{2}\left[H_2,\e^{-i[H_2,\cdot]\tau/2}(-iH_1)\right]U_{prod} \nonumber\\
&\ \ \  + \e^{-i[H_2,\cdot]\tau/2}(-iH_1)U'_{prod}+U'_{prod}\frac{-i}{2}H_2 \nonumber\\
&\ \ \ +(H_1+H_2)^2U\nonumber \\ 
&\implies \frac{d^2}{d\tau^2}\delta(\tau)\Big|_{\tau=0}=0
\end{align}

\begin{align}
&\frac{d^3}{d\tau^3}\delta(\tau)=-\frac{i}{2}H_2U''_{prod} \nonumber\\
&\ \ \ \ \ \ \ \ \ \  \ \ +\left(\frac{-i}{2}\right)^2\left[H_2,\left[H_2,\e^{-i[H_2,\cdot]\tau/2}(-iH_1)\right]\right]U_{prod}\nonumber\\
&\ \ \ \ \ \ \ \ \ \  \ \ \ +2\frac{-i}{2}\left[H_2,\e^{-i[H_2,\cdot]\tau/2}(-iH_1)\right]U'_{prod} \nonumber\\
&\ \ \ \ \ \ \ \ \ \  \ \ \ +\e^{-i[H_2,\cdot]\tau/2}(-iH_1)U''_{prod}+\frac{-i}{2}U''_{prod}H_2\nonumber\\
&\ \ \ \ \ \ \ \ \ \  \ \ \ -i(H_1+H_2)^3U \\
&\implies \frac{d^3}{d\tau^3}\delta(\tau)\Big|_{\tau=0}\nonumber\\
&\ \ \ \ \ \ \ \ \ \ \ \ =\frac{i}{2}H_2(H_1+H_2)^2+\frac{i}{4}\left[H_2,\left[H_2,H_1\right]\right]\nonumber\\
&\ \ \ \ \ \ \ \ \ \ \ \ \ \ \  +i[H_2,H_1](H_1+H_2)+i H_1(H_1+H_2)^2 \nonumber\\
&\ \ \ \ \ \ \ \ \ \ \ \ \ \ \  +\frac{i}{2}(H_1+H_2)^2H_2-i(H_1+H_2)^3 \nonumber\\
&\ \ \ \ \ \ \ \ \ \ \ \ =\frac{i}{4}[H_2,[H_1,H_2]]+\frac{i}{2}[H_1,[H_1,H_2]]
\end{align}
Therefore, the leading order of Trotter error per step is 
\begin{align}
&\delta(\tau)=\frac{1}{3!}\frac{d^3\delta}{d\tau^3}\Big|_{\tau=0} \tau^3+O(\tau^4)\nonumber\\
&=\left(\frac{i}{24}[H_2,[H_1,H_2]]+\frac{i}{12}[H_1,[H_1,H_2]]\right)\tau^3+O(\tau^4) \nonumber \\
&=\(\frac{i}{24}[H_1,[H_1,H_2]]+\frac{i}{24}[H,[H,H_2]]\)\tau^3+O(\tau^4). 
\end{align}

\section{Trotter error in observables}\label{app:obs_trot_err}
To compute Trotter error in observables, we first evaluate the operator difference in the Heisenberg picture, i.e.,
\begin{align}
&(U+\delta)^{-r}O(U+\delta)^r-U^{-r}OU^r\\
&= \left[U^{-r}+U^{-r+1}\Delta^{\dagger}\right]O\left[U^{r}+\Delta U^{r-1} \right]- U^{-r}OU^r\nonumber\\
&\ \ \  +O\left(\frac{t^{p+1}}{r^{p+1}}\right)\\
&=U^{-r}(U\Delta^{\dagger}O+O\Delta U^{-1})U^r + O\left(\frac{t^{p+1}}{r^{p+1}}\right) \\
&=U^{-r}(\Delta^{\dagger}O+O\Delta )U^r+ O\left(\frac{t^{p+2}}{r^{p+1}}\right),
\end{align}
where the last equation uses $U=I-i\tau H+O[\tau^2]$ and the fact that $\Delta$ is of order $t^{p+1}/r^p$. To evaluate the observable Trotter on a density matrix $\rho$, we have 
\begin{align}
\triangle\langle O\rangle_{\rho}&:=\Tr\left[(U+\delta)^{-r}O(U+\delta)^r\rho-U^{-r}OU^r\rho\right] \nonumber\\
&=\Tr\left[(\Delta^{\dagger}O+O\Delta)(U^r\rho U^{-r})\right]+ O\left(\frac{t^{p+2}}{r^{p+1}}\right).
\end{align}
For a Hamiltonian with two non-commuting blocks $H=H_1+H_2$, the error per step of a second-order product formula is
\begin{align}
\delta&=i\left(\frac{1}{24}[H_2,[H_1,H_2]]+\frac{1}{12}[H_1,[H_1,H_2]]\right)\frac{t^3}{r^3}+O\left(\frac{t^4}{r^4}\right) \nonumber\\
&:= i( H_{||}+ H_{\perp})\frac{t^3}{r^3}+O\left(\frac{t^4}{r^4}\right).
\end{align}
Notice that both terms $[H_2,[H_1,H_2]]$ and $[H_1,[H_1,H_2]]$ are Hermitian and therefore the commutant decomposition must have $H_{||}$ and $H_{\perp}$ being Hermitian---since complex conjugation leaves subspaces $\mathbf{H}_{||}$ and $\mathbf{H}_{\perp}$ invariant, $H_{||}$ and $H_{\perp}$ must be individually Hermitian when $H_{||}+H_{\perp}$ is. The cumulative error is
\begin{align}
&\Delta=iH_{||}\frac{t^3}{r^2}+\left[\ad^{-1}(H_{\perp})-\e^{-iHt}\ad^{-1}(H_{\perp})\e^{iHt}\right]\frac{t^2}{r^2}\nonumber\\
&\ \ \ \ \ \  +O\left(\frac{t^4}{r^3}\right),
\end{align}
and the Trotter error of the observable is
\begin{align}
&\triangle\langle O\rangle_{\rho}=\Tr\left[(\Delta^{\dagger}O+O\Delta)(U^r\rho U^{-r})\right]+ O\left(\frac{t^4}{r^3}\right)\nonumber \\
&=\Tr\Bigg[ \Big(i\left[O,H_{||}\right]\frac{t^3}{r^2} \nonumber \\
&\ \ \ \ \ \ \ \ + \big[O,\ad^{-1}(H_{\perp})-\e^{-iHt}\ad^{-1}(H_{\perp})\e^{iHt}\big]\frac{t^2}{r^2}\Big)\nonumber \\
&\ \ \ \ \ \ \ \  \times \left(\e^{-iHt}\rho \e^{iHt}\right) \Bigg]+O\left(\frac{t^4}{r^3}\right).
\end{align}

\bibliography{rf.bib}

\providecommand{\noopsort}[1]{}\providecommand{\singleletter}[1]{#1}%
\begin{thebibliography}{26}%
\makeatletter
\providecommand \@ifxundefined [1]{%
 \@ifx{#1\undefined}
}%
\providecommand \@ifnum [1]{%
 \ifnum #1\expandafter \@firstoftwo
 \else \expandafter \@secondoftwo
 \fi
}%
\providecommand \@ifx [1]{%
 \ifx #1\expandafter \@firstoftwo
 \else \expandafter \@secondoftwo
 \fi
}%
\providecommand \natexlab [1]{#1}%
\providecommand \enquote  [1]{``#1''}%
\providecommand \bibnamefont  [1]{#1}%
\providecommand \bibfnamefont [1]{#1}%
\providecommand \citenamefont [1]{#1}%
\providecommand \href@noop [0]{\@secondoftwo}%
\providecommand \href [0]{\begingroup \@sanitize@url \@href}%
\providecommand \@href[1]{\@@startlink{#1}\@@href}%
\providecommand \@@href[1]{\endgroup#1\@@endlink}%
\providecommand \@sanitize@url [0]{\catcode `\\12\catcode `\$12\catcode
  `\&12\catcode `\#12\catcode `\^12\catcode `\_12\catcode `\%12\relax}%
\providecommand \@@startlink[1]{}%
\providecommand \@@endlink[0]{}%
\providecommand \url  [0]{\begingroup\@sanitize@url \@url }%
\providecommand \@url [1]{\endgroup\@href {#1}{\urlprefix }}%
\providecommand \urlprefix  [0]{URL }%
\providecommand \Eprint [0]{\href }%
\providecommand \doibase [0]{https://doi.org/}%
\providecommand \selectlanguage [0]{\@gobble}%
\providecommand \bibinfo  [0]{\@secondoftwo}%
\providecommand \bibfield  [0]{\@secondoftwo}%
\providecommand \translation [1]{[#1]}%
\providecommand \BibitemOpen [0]{}%
\providecommand \bibitemStop [0]{}%
\providecommand \bibitemNoStop [0]{.\EOS\space}%
\providecommand \EOS [0]{\spacefactor3000\relax}%
\providecommand \BibitemShut  [1]{\csname bibitem#1\endcsname}%
\let\auto@bib@innerbib\@empty
\bibitem [{\citenamefont {Lloyd}(1996)}]{Lloyd1996}%
  \BibitemOpen
  \bibfield  {author} {\bibinfo {author} {\bibfnamefont {S.}~\bibnamefont
  {Lloyd}},\ }\bibfield  {title} {\bibinfo {title} {Universal quantum
  simulators},\ }\href {https://doi.org/10.1126/science.273.5278.1073}
  {\bibfield  {journal} {\bibinfo  {journal} {Science}\ }\textbf {\bibinfo
  {volume} {273}},\ \bibinfo {pages} {1073} (\bibinfo {year}
  {1996})}\BibitemShut {NoStop}%
\bibitem [{\citenamefont {Trotter}(1959)}]{trotter}%
  \BibitemOpen
  \bibfield  {author} {\bibinfo {author} {\bibfnamefont {H.~F.}\ \bibnamefont
  {Trotter}},\ }\bibfield  {title} {\bibinfo {title} {On the product of
  semi-groups of operators},\ }\href {http://www.jstor.org/stable/2033649}
  {\bibfield  {journal} {\bibinfo  {journal} {Proceedings of the American
  Mathematical Society}\ }\textbf {\bibinfo {volume} {10}},\ \bibinfo {pages}
  {545} (\bibinfo {year} {1959})}\BibitemShut {NoStop}%
\bibitem [{\citenamefont {Berry}\ \emph {et~al.}(2015)\citenamefont {Berry},
  \citenamefont {Childs}, \citenamefont {Cleve}, \citenamefont {Kothari},\ and\
  \citenamefont {Somma}}]{LCU2015}%
  \BibitemOpen
  \bibfield  {author} {\bibinfo {author} {\bibfnamefont {D.~W.}\ \bibnamefont
  {Berry}}, \bibinfo {author} {\bibfnamefont {A.~M.}\ \bibnamefont {Childs}},
  \bibinfo {author} {\bibfnamefont {R.}~\bibnamefont {Cleve}}, \bibinfo
  {author} {\bibfnamefont {R.}~\bibnamefont {Kothari}},\ and\ \bibinfo {author}
  {\bibfnamefont {R.~D.}\ \bibnamefont {Somma}},\ }\bibfield  {title} {\bibinfo
  {title} {Simulating hamiltonian dynamics with a truncated taylor series},\
  }\href {https://doi.org/10.1103/PhysRevLett.114.090502} {\bibfield  {journal}
  {\bibinfo  {journal} {Phys. Rev. Lett.}\ }\textbf {\bibinfo {volume} {114}},\
  \bibinfo {pages} {090502} (\bibinfo {year} {2015})}\BibitemShut {NoStop}%
\bibitem [{\citenamefont {Low}\ and\ \citenamefont {Chuang}(2017)}]{QSP2017}%
  \BibitemOpen
  \bibfield  {author} {\bibinfo {author} {\bibfnamefont {G.~H.}\ \bibnamefont
  {Low}}\ and\ \bibinfo {author} {\bibfnamefont {I.~L.}\ \bibnamefont
  {Chuang}},\ }\bibfield  {title} {\bibinfo {title} {Optimal hamiltonian
  simulation by quantum signal processing},\ }\href
  {https://doi.org/10.1103/PhysRevLett.118.010501} {\bibfield  {journal}
  {\bibinfo  {journal} {Phys. Rev. Lett.}\ }\textbf {\bibinfo {volume} {118}},\
  \bibinfo {pages} {010501} (\bibinfo {year} {2017})}\BibitemShut {NoStop}%
\bibitem [{\citenamefont {Childs}\ \emph {et~al.}(2018)\citenamefont {Childs},
  \citenamefont {Maslov}, \citenamefont {Nam}, \citenamefont {Ross},\ and\
  \citenamefont {Su}}]{Childs2018}%
  \BibitemOpen
  \bibfield  {author} {\bibinfo {author} {\bibfnamefont {A.~M.}\ \bibnamefont
  {Childs}}, \bibinfo {author} {\bibfnamefont {D.}~\bibnamefont {Maslov}},
  \bibinfo {author} {\bibfnamefont {Y.}~\bibnamefont {Nam}}, \bibinfo {author}
  {\bibfnamefont {N.~J.}\ \bibnamefont {Ross}},\ and\ \bibinfo {author}
  {\bibfnamefont {Y.}~\bibnamefont {Su}},\ }\bibfield  {title} {\bibinfo
  {title} {Toward the first quantum simulation with quantum speedup},\ }\href
  {https://doi.org/10.1073/pnas.1801723115} {\bibfield  {journal} {\bibinfo
  {journal} {Proceedings of the National Academy of Sciences}\ }\textbf
  {\bibinfo {volume} {115}},\ \bibinfo {pages} {9456} (\bibinfo {year}
  {2018})}\BibitemShut {NoStop}%
\bibitem [{\citenamefont {Suzuki}(1990)}]{SUZUKI1990319}%
  \BibitemOpen
  \bibfield  {author} {\bibinfo {author} {\bibfnamefont {M.}~\bibnamefont
  {Suzuki}},\ }\bibfield  {title} {\bibinfo {title} {Fractal decomposition of
  exponential operators with applications to many-body theories and monte carlo
  simulations},\ }\href
  {https://doi.org/https://doi.org/10.1016/0375-9601(90)90962-N} {\bibfield
  {journal} {\bibinfo  {journal} {Physics Letters A}\ }\textbf {\bibinfo
  {volume} {146}},\ \bibinfo {pages} {319} (\bibinfo {year}
  {1990})}\BibitemShut {NoStop}%
\bibitem [{\citenamefont {Hatano}\ and\ \citenamefont
  {Suzuki}(2005)}]{Hatano2005}%
  \BibitemOpen
  \bibfield  {author} {\bibinfo {author} {\bibfnamefont {N.}~\bibnamefont
  {Hatano}}\ and\ \bibinfo {author} {\bibfnamefont {M.}~\bibnamefont
  {Suzuki}},\ }\bibinfo {title} {Finding exponential product formulas of higher
  orders},\ in\ \href {https://doi.org/10.1007/11526216_2} {\emph {\bibinfo
  {booktitle} {Quantum Annealing and Other Optimization Methods}}},\ \bibinfo
  {editor} {edited by\ \bibinfo {editor} {\bibfnamefont {A.}~\bibnamefont
  {Das}}\ and\ \bibinfo {editor} {\bibfnamefont {B.}~\bibnamefont
  {K.~Chakrabarti}}}\ (\bibinfo  {publisher} {Springer Berlin Heidelberg},\
  \bibinfo {address} {Berlin, Heidelberg},\ \bibinfo {year} {2005})\ pp.\
  \bibinfo {pages} {37--68}\BibitemShut {NoStop}%
\bibitem [{\citenamefont {Childs}\ \emph {et~al.}(2019)\citenamefont {Childs},
  \citenamefont {Ostrander},\ and\ \citenamefont
  {Su}}]{Childs2019fasterquantum}%
  \BibitemOpen
  \bibfield  {author} {\bibinfo {author} {\bibfnamefont {A.~M.}\ \bibnamefont
  {Childs}}, \bibinfo {author} {\bibfnamefont {A.}~\bibnamefont {Ostrander}},\
  and\ \bibinfo {author} {\bibfnamefont {Y.}~\bibnamefont {Su}},\ }\bibfield
  {title} {\bibinfo {title} {Faster quantum simulation by randomization},\
  }\href {https://doi.org/10.22331/q-2019-09-02-182} {\bibfield  {journal}
  {\bibinfo  {journal} {{Quantum}}\ }\textbf {\bibinfo {volume} {3}},\ \bibinfo
  {pages} {182} (\bibinfo {year} {2019})}\BibitemShut {NoStop}%
\bibitem [{\citenamefont {Low}\ \emph {et~al.}(2019)\citenamefont {Low},
  \citenamefont {Kliuchnikov},\ and\ \citenamefont {Wiebe}}]{low2019}%
  \BibitemOpen
  \bibfield  {author} {\bibinfo {author} {\bibfnamefont {G.~H.}\ \bibnamefont
  {Low}}, \bibinfo {author} {\bibfnamefont {V.}~\bibnamefont {Kliuchnikov}},\
  and\ \bibinfo {author} {\bibfnamefont {N.}~\bibnamefont {Wiebe}},\ }\href
  {https://arxiv.org/abs/1907.11679} {\bibinfo {title} {Well-conditioned
  multiproduct hamiltonian simulation}} (\bibinfo {year} {2019}),\ \Eprint
  {https://arxiv.org/abs/1907.11679} {arXiv:1907.11679 [quant-ph]} \BibitemShut
  {NoStop}%
\bibitem [{\citenamefont {Carrera~Vazquez}\ \emph {et~al.}(2023)\citenamefont
  {Carrera~Vazquez}, \citenamefont {Egger}, \citenamefont {Ochsner},\ and\
  \citenamefont {Woerner}}]{Carrera2023}%
  \BibitemOpen
  \bibfield  {author} {\bibinfo {author} {\bibfnamefont {A.}~\bibnamefont
  {Carrera~Vazquez}}, \bibinfo {author} {\bibfnamefont {D.~J.}\ \bibnamefont
  {Egger}}, \bibinfo {author} {\bibfnamefont {D.}~\bibnamefont {Ochsner}},\
  and\ \bibinfo {author} {\bibfnamefont {S.}~\bibnamefont {Woerner}},\
  }\bibfield  {title} {\bibinfo {title} {Well-conditioned multi-product
  formulas for hardware-friendly {H}amiltonian simulation},\ }\href
  {https://doi.org/10.22331/q-2023-07-25-1067} {\bibfield  {journal} {\bibinfo
  {journal} {{Quantum}}\ }\textbf {\bibinfo {volume} {7}},\ \bibinfo {pages}
  {1067} (\bibinfo {year} {2023})}\BibitemShut {NoStop}%
\bibitem [{\citenamefont {Mc~Keever}\ and\ \citenamefont
  {Lubasch}(2023)}]{lubi2023}%
  \BibitemOpen
  \bibfield  {author} {\bibinfo {author} {\bibfnamefont {C.}~\bibnamefont
  {Mc~Keever}}\ and\ \bibinfo {author} {\bibfnamefont {M.}~\bibnamefont
  {Lubasch}},\ }\bibfield  {title} {\bibinfo {title} {Classically optimized
  hamiltonian simulation},\ }\href
  {https://doi.org/10.1103/PhysRevResearch.5.023146} {\bibfield  {journal}
  {\bibinfo  {journal} {Phys. Rev. Res.}\ }\textbf {\bibinfo {volume} {5}},\
  \bibinfo {pages} {023146} (\bibinfo {year} {2023})}\BibitemShut {NoStop}%
\bibitem [{\citenamefont {Childs}\ \emph {et~al.}(2021)\citenamefont {Childs},
  \citenamefont {Su}, \citenamefont {Tran}, \citenamefont {Wiebe},\ and\
  \citenamefont {Zhu}}]{PhysRevX.11.011020}%
  \BibitemOpen
  \bibfield  {author} {\bibinfo {author} {\bibfnamefont {A.~M.}\ \bibnamefont
  {Childs}}, \bibinfo {author} {\bibfnamefont {Y.}~\bibnamefont {Su}}, \bibinfo
  {author} {\bibfnamefont {M.~C.}\ \bibnamefont {Tran}}, \bibinfo {author}
  {\bibfnamefont {N.}~\bibnamefont {Wiebe}},\ and\ \bibinfo {author}
  {\bibfnamefont {S.}~\bibnamefont {Zhu}},\ }\bibfield  {title} {\bibinfo
  {title} {Theory of trotter error with commutator scaling},\ }\href
  {https://doi.org/10.1103/PhysRevX.11.011020} {\bibfield  {journal} {\bibinfo
  {journal} {Phys. Rev. X}\ }\textbf {\bibinfo {volume} {11}},\ \bibinfo
  {pages} {011020} (\bibinfo {year} {2021})}\BibitemShut {NoStop}%
\bibitem [{\citenamefont {Childs}\ and\ \citenamefont
  {Su}(2019)}]{PhysRevLett.123.050503}%
  \BibitemOpen
  \bibfield  {author} {\bibinfo {author} {\bibfnamefont {A.~M.}\ \bibnamefont
  {Childs}}\ and\ \bibinfo {author} {\bibfnamefont {Y.}~\bibnamefont {Su}},\
  }\bibfield  {title} {\bibinfo {title} {Nearly optimal lattice simulation by
  product formulas},\ }\href {https://doi.org/10.1103/PhysRevLett.123.050503}
  {\bibfield  {journal} {\bibinfo  {journal} {Phys. Rev. Lett.}\ }\textbf
  {\bibinfo {volume} {123}},\ \bibinfo {pages} {050503} (\bibinfo {year}
  {2019})}\BibitemShut {NoStop}%
\bibitem [{\citenamefont {Tran}\ \emph {et~al.}(2020)\citenamefont {Tran},
  \citenamefont {Chu}, \citenamefont {Su}, \citenamefont {Childs},\ and\
  \citenamefont {Gorshkov}}]{PhysRevLett.124.220502}%
  \BibitemOpen
  \bibfield  {author} {\bibinfo {author} {\bibfnamefont {M.~C.}\ \bibnamefont
  {Tran}}, \bibinfo {author} {\bibfnamefont {S.-K.}\ \bibnamefont {Chu}},
  \bibinfo {author} {\bibfnamefont {Y.}~\bibnamefont {Su}}, \bibinfo {author}
  {\bibfnamefont {A.~M.}\ \bibnamefont {Childs}},\ and\ \bibinfo {author}
  {\bibfnamefont {A.~V.}\ \bibnamefont {Gorshkov}},\ }\bibfield  {title}
  {\bibinfo {title} {Destructive error interference in product-formula lattice
  simulation},\ }\href {https://doi.org/10.1103/PhysRevLett.124.220502}
  {\bibfield  {journal} {\bibinfo  {journal} {Phys. Rev. Lett.}\ }\textbf
  {\bibinfo {volume} {124}},\ \bibinfo {pages} {220502} (\bibinfo {year}
  {2020})}\BibitemShut {NoStop}%
\bibitem [{\citenamefont {Layden}(2022)}]{PhysRevLett.128.210501}%
  \BibitemOpen
  \bibfield  {author} {\bibinfo {author} {\bibfnamefont {D.}~\bibnamefont
  {Layden}},\ }\bibfield  {title} {\bibinfo {title} {First-order trotter error
  from a second-order perspective},\ }\href
  {https://doi.org/10.1103/PhysRevLett.128.210501} {\bibfield  {journal}
  {\bibinfo  {journal} {Phys. Rev. Lett.}\ }\textbf {\bibinfo {volume} {128}},\
  \bibinfo {pages} {210501} (\bibinfo {year} {2022})}\BibitemShut {NoStop}%
\bibitem [{\citenamefont {Heyl}\ \emph {et~al.}(2019)\citenamefont {Heyl},
  \citenamefont {Hauke},\ and\ \citenamefont
  {Zoller}}]{doi:10.1126/sciadv.aau8342}%
  \BibitemOpen
  \bibfield  {author} {\bibinfo {author} {\bibfnamefont {M.}~\bibnamefont
  {Heyl}}, \bibinfo {author} {\bibfnamefont {P.}~\bibnamefont {Hauke}},\ and\
  \bibinfo {author} {\bibfnamefont {P.}~\bibnamefont {Zoller}},\ }\bibfield
  {title} {\bibinfo {title} {Quantum localization bounds trotter errors in
  digital quantum simulation},\ }\href {https://doi.org/10.1126/sciadv.aau8342}
  {\bibfield  {journal} {\bibinfo  {journal} {Science Advances}\ }\textbf
  {\bibinfo {volume} {5}},\ \bibinfo {pages} {eaau8342} (\bibinfo {year}
  {2019})}\BibitemShut {NoStop}%
\bibitem [{\citenamefont {Kivlichan}\ \emph {et~al.}(2020)\citenamefont
  {Kivlichan}, \citenamefont {Gidney}, \citenamefont {Berry}, \citenamefont
  {Wiebe}, \citenamefont {McClean}, \citenamefont {Sun}, \citenamefont {Jiang},
  \citenamefont {Rubin}, \citenamefont {Fowler}, \citenamefont {Aspuru-Guzik},
  \citenamefont {Neven},\ and\ \citenamefont
  {Babbush}}]{Kivlichan2020improvedfault}%
  \BibitemOpen
  \bibfield  {author} {\bibinfo {author} {\bibfnamefont {I.~D.}\ \bibnamefont
  {Kivlichan}}, \bibinfo {author} {\bibfnamefont {C.}~\bibnamefont {Gidney}},
  \bibinfo {author} {\bibfnamefont {D.~W.}\ \bibnamefont {Berry}}, \bibinfo
  {author} {\bibfnamefont {N.}~\bibnamefont {Wiebe}}, \bibinfo {author}
  {\bibfnamefont {J.}~\bibnamefont {McClean}}, \bibinfo {author} {\bibfnamefont
  {W.}~\bibnamefont {Sun}}, \bibinfo {author} {\bibfnamefont {Z.}~\bibnamefont
  {Jiang}}, \bibinfo {author} {\bibfnamefont {N.}~\bibnamefont {Rubin}},
  \bibinfo {author} {\bibfnamefont {A.}~\bibnamefont {Fowler}}, \bibinfo
  {author} {\bibfnamefont {A.}~\bibnamefont {Aspuru-Guzik}}, \bibinfo {author}
  {\bibfnamefont {H.}~\bibnamefont {Neven}},\ and\ \bibinfo {author}
  {\bibfnamefont {R.}~\bibnamefont {Babbush}},\ }\bibfield  {title} {\bibinfo
  {title} {Improved {F}ault-{T}olerant {Q}uantum {S}imulation of
  {C}ondensed-{P}hase {C}orrelated {E}lectrons via {T}rotterization},\ }\href
  {https://doi.org/10.22331/q-2020-07-16-296} {\bibfield  {journal} {\bibinfo
  {journal} {{Quantum}}\ }\textbf {\bibinfo {volume} {4}},\ \bibinfo {pages}
  {296} (\bibinfo {year} {2020})}\BibitemShut {NoStop}%
\bibitem [{Note1()}]{Note1}%
  \BibitemOpen
  \bibinfo {note} {Observe that this is a geometric series applied to the
  operator $\protect \text {ad}_H$, which is a diagonal matrix in $H$'s
  eigen-operator representation $|i\rangle \langle j|$ (Appendix~\ref
  {app:commutant_deomp}). The geometric series is then individually applied to
  each diagonal element.}\BibitemShut {Stop}%
\bibitem [{\citenamefont {Haah}\ \emph {et~al.}(2023)\citenamefont {Haah},
  \citenamefont {Hastings}, \citenamefont {Kothari},\ and\ \citenamefont
  {Low}}]{doi:10.1137/18M1231511}%
  \BibitemOpen
  \bibfield  {author} {\bibinfo {author} {\bibfnamefont {J.}~\bibnamefont
  {Haah}}, \bibinfo {author} {\bibfnamefont {M.~B.}\ \bibnamefont {Hastings}},
  \bibinfo {author} {\bibfnamefont {R.}~\bibnamefont {Kothari}},\ and\ \bibinfo
  {author} {\bibfnamefont {G.~H.}\ \bibnamefont {Low}},\ }\bibfield  {title}
  {\bibinfo {title} {Quantum algorithm for simulating real time evolution of
  lattice hamiltonians},\ }\href {https://doi.org/10.1137/18M1231511}
  {\bibfield  {journal} {\bibinfo  {journal} {SIAM Journal on Computing}\
  }\textbf {\bibinfo {volume} {52}},\ \bibinfo {pages} {FOCS18} (\bibinfo
  {year} {2023})}\BibitemShut {NoStop}%
\bibitem [{Note2()}]{Note2}%
  \BibitemOpen
  \bibinfo {note} {Setting the right-hand side of Eq.~\protect \eqref
  {existingbound} as 0.01 the required total accuracy, and using $\tau =t/r$
  with total time $t=400$, we deduce the required number of steps is $\protect
  \sqrt {\protect \frac {t^3}{0.01}\(\protect \frac {1}{12}\left |\left
  |[H_1,[H_1,H_2]]\right |\right |+\protect \frac {1}{24}\left |\left
  |[H_2,[H_2,H_1]]\right |\right |\)}$ resulting in 545244 steps required using
  exisiting bound Eq.~\protect \eqref {existingbound}. Similarly, setting the
  right-hand side of Eq.~\protect \eqref {tighterbound1} as 0.01 the required
  accuracy, we deduce the required steps as $\protect \sqrt {\protect \frac
  {1}{0.01}\(\protect \frac {1}{24}\left |\left |[H_1,[H_1,H_2]]\right |\right
  |t^3+\protect \frac {1}{12}\left |\left |[H_1,H_2]\right |\right
  |t^2\)}$}\BibitemShut {NoStop}%
\bibitem [{Note3()}]{Note3}%
  \BibitemOpen
  \bibinfo {note} {To clarify, the value of the $O(\tau ^2)$ term is a
  complicated combination of many oscillations and it could look linear in a
  particular period under certain circumstances, but we call it constant-in-t
  in the sense that its overall effect is oscillation that does not grow with
  time asymptotically}\BibitemShut {NoStop}%
\bibitem [{\citenamefont {Weidinger}\ and\ \citenamefont
  {Knap}(2017)}]{WeidingerKnap2017}%
  \BibitemOpen
  \bibfield  {author} {\bibinfo {author} {\bibfnamefont {S.~A.}\ \bibnamefont
  {Weidinger}}\ and\ \bibinfo {author} {\bibfnamefont {M.}~\bibnamefont
  {Knap}},\ }\bibfield  {title} {\bibinfo {title} {Floquet prethermalization
  and regimes of heating in a periodically driven, interacting quantum
  system},\ }\href {https://doi.org/10.1038/srep45382} {\bibfield  {journal}
  {\bibinfo  {journal} {Scientific Reports}\ }\textbf {\bibinfo {volume} {7}},\
  \bibinfo {pages} {45382} (\bibinfo {year} {2017})}\BibitemShut {NoStop}%
\bibitem [{\citenamefont {Ho}\ \emph {et~al.}(2023)\citenamefont {Ho},
  \citenamefont {Mori}, \citenamefont {Abanin},\ and\ \citenamefont {{Dalla
  Torre}}}]{HO2023169297}%
  \BibitemOpen
  \bibfield  {author} {\bibinfo {author} {\bibfnamefont {W.~W.}\ \bibnamefont
  {Ho}}, \bibinfo {author} {\bibfnamefont {T.}~\bibnamefont {Mori}}, \bibinfo
  {author} {\bibfnamefont {D.~A.}\ \bibnamefont {Abanin}},\ and\ \bibinfo
  {author} {\bibfnamefont {E.~G.}\ \bibnamefont {{Dalla Torre}}},\ }\bibfield
  {title} {\bibinfo {title} {Quantum and classical floquet prethermalization},\
  }\href {https://doi.org/https://doi.org/10.1016/j.aop.2023.169297} {\bibfield
   {journal} {\bibinfo  {journal} {Annals of Physics}\ }\textbf {\bibinfo
  {volume} {454}},\ \bibinfo {pages} {169297} (\bibinfo {year}
  {2023})}\BibitemShut {NoStop}%
\bibitem [{\citenamefont {Kuwahara}\ \emph {et~al.}(2016)\citenamefont
  {Kuwahara}, \citenamefont {Mori},\ and\ \citenamefont
  {Saito}}]{KUWAHARA201696}%
  \BibitemOpen
  \bibfield  {author} {\bibinfo {author} {\bibfnamefont {T.}~\bibnamefont
  {Kuwahara}}, \bibinfo {author} {\bibfnamefont {T.}~\bibnamefont {Mori}},\
  and\ \bibinfo {author} {\bibfnamefont {K.}~\bibnamefont {Saito}},\ }\bibfield
   {title} {\bibinfo {title} {Floquet–magnus theory and generic transient
  dynamics in periodically driven many-body quantum systems},\ }\href
  {https://doi.org/https://doi.org/10.1016/j.aop.2016.01.012} {\bibfield
  {journal} {\bibinfo  {journal} {Annals of Physics}\ }\textbf {\bibinfo
  {volume} {367}},\ \bibinfo {pages} {96} (\bibinfo {year} {2016})}\BibitemShut
  {NoStop}%
\bibitem [{\citenamefont {Mori}\ \emph {et~al.}(2016)\citenamefont {Mori},
  \citenamefont {Kuwahara},\ and\ \citenamefont
  {Saito}}]{PhysRevLett.116.120401}%
  \BibitemOpen
  \bibfield  {author} {\bibinfo {author} {\bibfnamefont {T.}~\bibnamefont
  {Mori}}, \bibinfo {author} {\bibfnamefont {T.}~\bibnamefont {Kuwahara}},\
  and\ \bibinfo {author} {\bibfnamefont {K.}~\bibnamefont {Saito}},\ }\bibfield
   {title} {\bibinfo {title} {Rigorous bound on energy absorption and generic
  relaxation in periodically driven quantum systems},\ }\href
  {https://doi.org/10.1103/PhysRevLett.116.120401} {\bibfield  {journal}
  {\bibinfo  {journal} {Phys. Rev. Lett.}\ }\textbf {\bibinfo {volume} {116}},\
  \bibinfo {pages} {120401} (\bibinfo {year} {2016})}\BibitemShut {NoStop}%
\bibitem [{Note4()}]{Note4}%
  \BibitemOpen
  \bibinfo {note} {The analysis holds well for $r||\delta ||\ll 1$. Since
  $r||\delta ||=O(\tau ^2t)$ for second-order PF, under a fixed step size $\tau
  =t/r$, it implies $t\ll \protect \frac {1}{\tau ^2}$.}\BibitemShut {Stop}%
\end{thebibliography}%

\end{document}